\documentclass{sig-alternate}

\usepackage[paper=letterpaper,width=7in,height=9.25in,centering]{geometry}
\usepackage[all,arc]{xy} \input xy

\begin{document}
%
\conferenceinfo{ISSAC 2010,}{25--28 July 2010, Munich, Germany.}
\CopyrightYear{2010}
\crdata{978-1-4503-0150-3/10/0007}

\title{Polynomial integration on regions defined by a triangle and a conic}

\numberofauthors{1}
\author{
%
%
\alignauthor
David Sevilla and Daniel Wachsmuth\\ 
       \affaddr{Johann Radon Institute for Computational and Applied Mathematics (RICAM)}\\
       \affaddr{Austrian Academy of Sciences}\\
       \affaddr{Altenbergerstrasse 69}\\
       \affaddr{A-4040 Linz, Austria}\\
       \email{david.sevilla@oeaw.ac.at , daniel.wachsmuth@oeaw.ac.at}
}

\date{23 January 2010}

\newtheorem{theorem}{Theorem}[section]
\newtheorem{lemma}[theorem]{Lemma}
\newtheorem{proposition}[theorem]{Proposition}
\newtheorem{algorithm}[theorem]{Algorithm}
\newtheorem{corollary}[theorem]{Corollary}
\newtheorem{conjecture}[theorem]{Conjecture}
\newtheorem{definition}[theorem]{Definition}
\newtheorem{example}[theorem]{Example}
\newtheorem{remark}[theorem]{Remark}
\newtheorem{assumption}[theorem]{Assumption}

\maketitle

\begin{abstract}

We present an efficient solution to the following problem, of relevance in a numerical optimization scheme: calculation of integrals of the type
\[
  \iint_{T \cap \{f\ge0\}} \phi_1\phi_2 \, dx\,dy
\]
for quadratic polynomials $f,\phi_1,\phi_2$ on a plane triangle $T$. The naive approach would involve consideration of the many possible shapes of $T\cap\{f\geq0\}$ (possibly after a convenient transformation) and parameterizing its border, in order to integrate the variables separately. Our solution involves partitioning the triangle into smaller triangles on which integration is much simpler.
\end{abstract}




\section{Introduction}

This article presents a symbolic solution to a problem of relevance in a numerical optimization scheme: the numerical solution of optimal control problems with partial differential equations as constraint requires to discretize the problem, i.e.\@ to solve finite-dimensional approximations, see e.g.\@ \cite{tro05}. When applying the variational discretization concept \cite{hin05}, the following problem arises: integrals of the type
\begin{equation}\label{eq01}
  \iint_{T \cap \{f\ge0\}} g(x,y) \, dx\,dy \,, \qquad g=\phi_1\cdot\phi_2
\end{equation}
for quadratic polynomials $f,\phi_1,\phi_2$ on a plane triangle $T$ have to be evaluated accurately.
Up to now the variational discretization method was used only for degree $1$, i.e.\@ where the function $f$ defining the integral region in \eqref{eq01} is a polynomial of degree $1$. Using polynomials with higher order gives better approximation results, see Theorem~\ref{theo55} below.

The naive approach to compute the integral \eqref{eq01} would involve consideration of the many possible shapes of $T\cap\{f\geq0\}$ and parameterizing its border, in order to integrate the variables separately. This suffers from some computational difficulties as we show below.

\begin{example}\label{example: ellipse partially in triangle}
  Suppose that one side of the triangle lies on a horizontal line. Consider the situation where the region of integration is the part of the interior of an ellipse in the triangle, as in the figure.
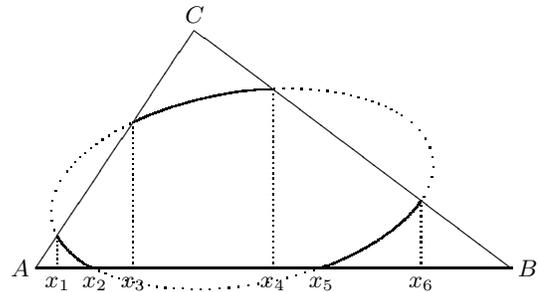
\[\xy
  0; <6pt,0pt>: 
  (0,0)*{}="A"; "A"+(-1,0)*{A};
  (30,0)*{}="B"; "B"+(1,0)*{B};
  (10,15)*{}="C"; "C"+(0,1)*{C};
  "A"; "B", **\dir{-}; "C", **\dir{-}; "A", **\dir{-};
  {
    (13,5)="c"; "c"+(6,1): 
    "c",{\ellipse(2,1):a(0),^,=:a(360){.}};
  };
  (1.35,0)="x1"; p+(0,2)="y1"; **\dir{.}; "x1"-(0,1)*{x_1};
  (3.6,0)="x2"; "x2"+(0.1,-1)*{x_2};
  (6.13,0)="x3"; "x3"+(0,9.2)="y3"; **\dir{.}; "x3"-(0,1)*{x_3};
  (14.93,0)="x4"; "x4"+(0,11.3)="y4"; **\dir{.}; "x4"-(0,1)*{x_4};
  (17.9,0)="x5"; "x5"+(0.1,-1)*{x_5};
  (24.28,0)="x6"; "x6"+(0,4.26)="y6"; **\dir{.}; "x6"-(0,1)*{x_6};
  "x2"; "y1", {\ellipse:a(70),_:a(30){}};
  "y3"; "y4", {\ellipse:a(117),_:a(90){}};
  "x5"; "y6", {\ellipse:a(-73),^:a(-32){}};
\endxy\]

We see that we could calculate the integral as the sum of five integrals on domains perpendicular to the $x$ axis. For example the first one is
\[
  \int_{x_1}^{x_2} \left( \int_{c_-(x)}^{l_{AC}(x)} g(x,y) \, dy \right) dx
\]
where $y=l_{AC}(x)$ is the equation of the line $AC$ and $y=c_-(x)$ is the equation of the lower part of the ellipse. Note that we need to parameterize the ellipse (this will involve at best a square root or trigonometric functions) and calculate the $x$-coordinates of the relevant points, which also involve square roots. The value of the inner integral will be given by a formula of which an antiderivative must be computed then. The resulting formula is far from simple.

An alternative would be to apply an affine transformation so that the ellipse becomes a circle centered at the origin, followed by a change to polar coordinates. This does not make the integral significantly easier to compute.

And of course, here we use our knowledge of the relative position of the ellipse and the triangle, as in the figure; the possible relative positions of a conic and a triangle are many, and to discern them is not trivial.
\end{example}

In contrast, in the following particular case we obtain a simple formula.

\begin{example}\label{example: full triangle}
  Let $A=(0,0)$, $B=(1,0)$, $C=(1,1)$, and assume that $f\geq0$ on the triangle $T = ABC$. If $g(x,y)=\displaystyle\sum_{i+j\leq4} b_{ij}x^iy^j$ then the integral becomes

  \begin{multline} \label{eq: full triangle integration formula}
    \iint_T g(x,y) \, dx\,dy = \int_0^1 \left( \int_0^x g(x,y) \, dy \right) dx = \\
    = \sum_{i+j\leq4} \frac{b_{ij}}{(j+1)(i+j+2)}
  \end{multline}

  For a general triangle $T$, one applies an affine transformation which brings the vertices to the above points. The resulting integrand is a polynomial of the same degree, and only a constant factor is introduced by the substitution formula.
\end{example}

Our solution involves partitioning $T$ into smaller triangles on which integration is much simpler. The result is a decision tree and several relatively simple explicit formulas, which form Algorithm \ref{algorithm: main algorithm}. The particular nature of $g$, beyond it being a polynomial, will be immaterial. Our method could in principle be adapted for larger values of $\deg f$, although it may become too complicated for practical uses even for degree 3. Besides, only the quadratic case is relevant for the context in which this problem arose. It is important to point out that an implementation for floating point arithmetic would need a more detailed treatment, see the end of Section \ref{section: application}.

We describe the subdivision method in Section \ref{section: triangular subdivision}, the integration in the base cases in Section \ref{section: base case integration}, the complete algorithm in Section \ref{section: the algorithm}, and a description of the application to optimization in Section \ref{section: application} which includes some comments on the practical implementation of our algorithm.

\section{Triangular subdivision}\label{section: triangular subdivision}

Our idea is to reduce the number of intersections between the curve $f=0$ and the sides of the triangle, by cutting the triangle into pieces until we reach some base cases that we establish below. For those cases the integration will be much simpler than in Example~\ref{example: ellipse partially in triangle}. We leave for later the case where the conic $f=0$ is degenerate (two lines, either intersecting, parallel, or coincident; one point; the empty set). Note that the type of a conic can be determined quickly by inspection of the equation.

First we introduce some nomenclature.

\begin{definition}
  Fix a nonsingular conic. A segment is called \emph{free} (with respect to the conic) if it does not intersect it except possibly at the vertices of the segment.
\end{definition}

\begin{remark}\label{remark: intersections of conic with lines and segments}$ $
  \begin{enumerate}
    \item Any line or segment intersects any conic at most at two points.
    \item A segment joining two points of the conic is always free.
  \end{enumerate}
\end{remark}

\begin{proof}
  Part 1 is a simple case of the weak B\'ezout's theorem \cite{walker1978}. Part 2 is a clear consequence of part 1.
\end{proof}

The calculation of the intersections of a segment with a given conic (and thus the determination of the freedom of the segment) is straightforward, and fast in practice.

The next definition encapsulates the base cases of our subdivision method.

\begin{definition}
  A triangle is called \emph{free} if either all its sides are free, or the intersection of its border with the conic is just one non-vertex point.
\end{definition}

Thus there are five types of free triangles: those with all sides free and 0, 1, 2 or 3 intersections at the vertices; and those with no vertex intersections and one side intersection.
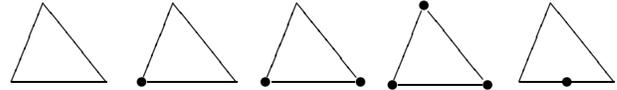
\[\xy
  0; <1.2pt,0pt>: 
  (0,0)*{\xy
    (0,0)*{}="A"; (30,0)*{}="B"; (10,25)*{}="C";
    "A"; "B"; **\dir{-};
    "B"; "C"; **\dir{-};
    "C"; "A"; **\dir{-};
    \endxy};
  (40,0)*{\xy
    (0,0)*{\bullet}="A"; (30,0)*{}="B"; (10,25)*{}="C";
    "A"; "B"; **\dir{-};
    "B"; "C"; **\dir{-};
    "C"; "A"; **\dir{-};
    \endxy};
  (80,0)*{\xy
    (0,0)*{\bullet}="A"; (30,0)*{\bullet}="B"; (10,25)*{}="C";
    "A"; "B"; **\dir{-};
    "B"; "C"; **\dir{-};
    "C"; "A"; **\dir{-};
    \endxy};
  (120,0)*{\xy
    (0,0)*{\bullet}="A"; (30,0)*{\bullet}="B"; (10,25)*{\bullet}="C";
    "A"; "B"; **\dir{-};
    "B"; "C"; **\dir{-};
    "C"; "A"; **\dir{-};
    \endxy};
  (160,0)*{\xy
    (0,0)*{}="A"; (30,0)*{}="B"; (10,25)*{}="C";
    "A"; "B"; **\dir{-}; ?(.5)*{\bullet}="AB"; 
    "B"; "C"; **\dir{-};
    "C"; "A"; **\dir{-};
    \endxy};
\endxy\]

The rest of the section describes how to divide a given triangle so that all the pieces are free triangles. We proceed step by step in terms of the number of free sides.

\begin{lemma}\label{lemma: no free sides -> 7 free triangles}
  Every triangle with no free sides can be cut into seven free triangles.
\end{lemma}

\begin{proof}
  Each non-free side has one or two interior intersections with the conic. We draw the four cases and one solution for each (possible vertex intersections are irrelevant here, thus not drawn). All the small triangles can be proven free by noting that their sides are either free parts of the original sides, or segments connecting two intersections (thus free by part 2 of Remark~\ref{remark: intersections of conic with lines and segments}).

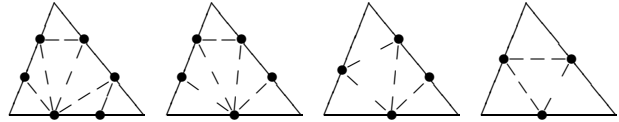
\[\xy
  0; <1.7pt,0pt>: 
  (0,0)*{\xy
    (0,0)*{}="A"; (30,0)*{}="B"; (10,25)*{}="C";
    "A"; "B"; **\dir{-}; ?(.33)*{\bullet}="ABB"; ?(.67)*{\bullet}="AAB"; 
    "B"; "C"; **\dir{-}; ?(.33)*{\bullet}="BCC"; ?(.67)*{\bullet}="BBC";
    "C"; "A"; **\dir{-}; ?(.33)*{\bullet}="AAC"; ?(.67)*{\bullet}="ACC";
    "AAB"; "AAC"; **\dir{--}; 
    "AAB"; "ACC"; **\dir{--};
    "AAB"; "BCC"; **\dir{--};
    "AAB"; "BBC"; **\dir{--};
    "ABB"; "BBC"; **\dir{--};
    "ACC"; "BCC"; **\dir{--};
    \endxy};
  (35,0)*{\xy
    (0,0)*{}="A"; (30,0)*{}="B"; (10,25)*{}="C";
    "A"; "B"; **\dir{-}; ?(.5)*{\bullet}="AB"; 
    "B"; "C"; **\dir{-}; ?(.33)*{\bullet}="BCC"; ?(.67)*{\bullet}="BBC";
    "C"; "A"; **\dir{-}; ?(.33)*{\bullet}="AAC"; ?(.67)*{\bullet}="ACC";
    "AB"; "AAC"; **\dir{--}; 
    "AB"; "ACC"; **\dir{--};
    "AB"; "BCC"; **\dir{--};
    "AB"; "BBC"; **\dir{--};
    "ACC"; "BCC"; **\dir{--};
    \endxy};
  (70,0)*{\xy
    (0,0)*{}="A"; (30,0)*{}="B"; (10,25)*{}="C";
    "A"; "B"; **\dir{-}; ?(.5)*{\bullet}="AB"; 
    "B"; "C"; **\dir{-}; ?(.33)*{\bullet}="BCC"; ?(.67)*{\bullet}="BBC";
    "C"; "A"; **\dir{-}; ?(.4)*{\bullet}="AC";
    "AB"; "AC"; **\dir{--}; 
    "AB"; "BCC"; **\dir{--};
    "AB"; "BBC"; **\dir{--};
    "AC"; "BCC"; **\dir{--};
    \endxy};
  (105,0)*{\xy
    (0,0)*{}="A"; (30,0)*{}="B"; (10,25)*{}="C";
    "A"; "B"; **\dir{-}; ?(.55)*{\bullet}="AB"; 
    "B"; "C"; **\dir{-}; ?(.5)*{\bullet}="BC";
    "C"; "A"; **\dir{-}; ?(.5)*{\bullet}="AC";
    "AB"; "AC"; **\dir{--}; 
    "AB"; "BC"; **\dir{--};
    "AC"; "BC"; **\dir{--};
    \endxy};
\endxy\]
\end{proof}

The next step is to consider non-free triangles with one free side. We introduce another useful term.

\begin{definition}
  A triangle is \emph{almost free} if exactly one of its sides is not free, and that side has only one intersection in its interior.
\end{definition}

\begin{remark}\label{remark: types of almost-free triangles}
  There are four types of almost-free triangles, depending on the vertex intersections.
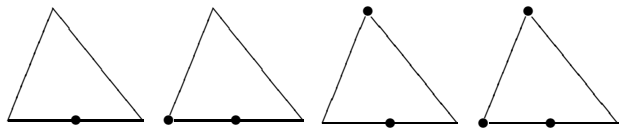
\[\xy
  0; <1.7pt,0pt>: 
  (0,0)*{\xy
    (0,0)*{}="A"; (30,0)*{}="B"; (10,25)*{}="C";
    "A"; "B"; **\dir{-}; ?(.5)*{\bullet}="AB"; 
    "B"; "C"; **\dir{-};
    "C"; "A"; **\dir{-};
    \endxy};
  (35,0)*{\xy
    (0,0)*{\bullet}="A"; (30,0)*{}="B"; (10,25)*{}="C";
    "A"; "B"; **\dir{-}; ?(.5)*{\bullet}="AB"; 
    "B"; "C"; **\dir{-};
    "C"; "A"; **\dir{-};
    \endxy};
  (70,0)*{\xy
    (0,0)*{}="A"; (30,0)*{}="B"; (10,25)*{\bullet}="C";
    "A"; "B"; **\dir{-}; ?(.5)*{\bullet}="AB"; 
    "B"; "C"; **\dir{-};
    "C"; "A"; **\dir{-};
    \endxy};
  (105,0)*{\xy
    (0,0)*{\bullet}="A"; (30,0)*{}="B"; (10,25)*{\bullet}="C";
    "A"; "B"; **\dir{-}; ?(.5)*{\bullet}="AB"; 
    "B"; "C"; **\dir{-};
    "C"; "A"; **\dir{-};
    \endxy};
\endxy\]
  Note that if a triangle is almost free and has no vertex intersections with the conic, then it is free (first case).
\end{remark}

\begin{lemma}\label{lemma: 1 free side -> 5 almost-free or free triangles}
  Every triangle with exactly one free side can be cut into five free or almost-free triangles, with zero or two of them being almost free.
\end{lemma}

\begin{proof}
  There are three cases depending on the number of interior intersections with the non-free sides: $2+2$, $2+1$ and $1+1$. The diagrams show how to cut the triangle in the three cases (possible vertex intersections, marked in white, make no difference). The numbers of free sides in each piece are indicated. As before, if a segment intersects a conic in its endpoints then it is free by part 2 of Remark~\ref{remark: intersections of conic with lines and segments}.
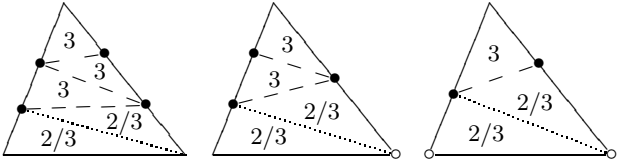
\[\xy
  0; <2.3pt,0pt>: 
  (0,0)*{\xy
    (0,0)*{}="A"; (30,0)*{}="B"; (10,25)*{}="C";
    "A"; "B"; **\dir{-}; 
    "B"; "C"; **\dir{-}; ?(.33)*{\bullet}="BCC"; ?(.67)*{\bullet}="BBC";
    "C"; "A"; **\dir{-}; ?(.3)*{\bullet}="AAC"; ?(.6)*{\bullet}="ACC";
    "B"; "AAC"; **\dir{.}; 
    "BBC"; "AAC"; **\dir{--};
    "BBC"; "ACC"; **\dir{--};
    "BCC"; "ACC"; **\dir{--};
    "C"; "ACC"; **\dir{}; ?(.5)*{}; "BCC"; **\dir{}; ?(.67)*{3}; 
    "BBC"; "ACC"; **\dir{}; ?(.5)*{}; "BCC"; **\dir{}; ?(.67)+(.5,.5)*{3};
    "BBC"; "ACC"; **\dir{}; ?(.5)*{}; "AAC"; **\dir{}; ?(.67)+(-1,.5)*{3};
    "BBC"; "B"; **\dir{}; ?(.5)*{}; "AAC"; **\dir{}; ?(.67)+(1,0)*{2/3};
    "B"; "A"; **\dir{}; ?(.5)*{}; "AAC"; **\dir{}; ?(.67)+(-2,0)*{2/3};
    \endxy};
  (35,0)*{\xy
    (0,0)*{}="A"; (30,0)*{\circ}="B"; (10,25)*{}="C";
    "A"; "B"; **\dir{-}; 
    "B"; "C"; **\dir{-}; ?(.5)*{\bullet}="BC";
    "C"; "A"; **\dir{-}; ?(.33)*{\bullet}="AAC"; ?(.67)*{\bullet}="ACC";
    "B"; "AAC"; **\dir{.}; 
    "BC"; "AAC"; **\dir{--};
    "BC"; "ACC"; **\dir{--};
    "C"; "ACC"; **\dir{}; ?(.5)*{}; "BC"; **\dir{}; ?(.67)*{3}; 
    "BC"; "ACC"; **\dir{}; ?(.5)*{}; "AAC"; **\dir{}; ?(.67)*{3};
    "BC"; "B"; **\dir{}; ?(.5)*{}; "AAC"; **\dir{}; ?(.67)*{2/3};
    "B"; "A"; **\dir{}; ?(.5)*{}; "AAC"; **\dir{}; ?(.67)+(-2,0)*{2/3};
    \endxy};
  (70,0)*{\xy
    (0,0)*{\circ}="A"; (30,0)*{\circ}="B"; (10,25)*{}="C";
    "A"; "B"; **\dir{-}; 
    "B"; "C"; **\dir{-}; ?(.4)*{\bullet}="BC";
    "C"; "A"; **\dir{-}; ?(.4)*{\bullet}="AC";
    "B"; "AC"; **\dir{.}; 
    "BC"; "AC"; **\dir{--};
    "C"; "AC"; **\dir{}; ?(.5)*{}; "BC"; **\dir{}; ?(.67)*{3}; 
    "BC"; "B"; **\dir{}; ?(.5)*{}; "AC"; **\dir{}; ?(.67)*{2/3};
    "B"; "A"; **\dir{}; ?(.5)*{}; "AC"; **\dir{}; ?(.67)+(-2,0)*{2/3};
    \endxy};
\endxy\]

In all cases, the only segment that may not be free is the lowest new one (dotted), and it can only have one interior intersection (by part 1 of Remark~\ref{remark: intersections of conic with lines and segments}, since one of its endpoints is in the conic already).
\end{proof}

Next, we consider the case when two sides are free.

\begin{lemma}\label{lemma: 2 free sides -> 4 almost-free or free triangles}
  Every triangle with exactly two free sides can be cut into four free or almost-free triangles. At least one of them is free, except if the original triangle is almost free.
\end{lemma}

\begin{proof}
  The non-free side has one or two interior intersections; in the former case, the triangle is almost free and we are finished. If it has two interior intersections, we join them with the opposite vertex and create two interior sides. There are three possibilities:
  \begin{enumerate}
    \item If there are no interior intersections in the new sides, the three pieces are free.
    \item If one of the new sides has one interior point, we obtain one free triangle and two almost-free triangles.
    \item If both new sides have one interior point each, with one extra cut we obtain one free triangle and three almost-free triangles.
  \end{enumerate}

  The dashed lines indicate the partitions described above.
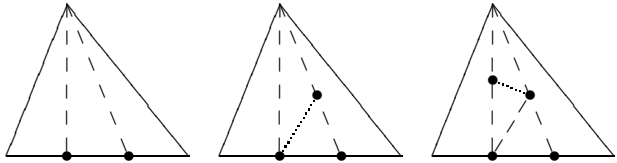
\[\xy
  0; <2.3pt,0pt>: 
  (0,0)*{\xy
    (0,0)*{}="A"; (30,0)*{}="B"; (10,25)*{}="C";
    "A"; "B"; **\dir{-}; ?(.33)*{\bullet}="ABB"; ?(.67)*{\bullet}="AAB"; 
    "B"; "C"; **\dir{-};
    "C"; "A"; **\dir{-};
    "C"; "AAB"; **\dir{--}; 
    "C"; "ABB"; **\dir{--};
    \endxy};
  (35,0)*{\xy
    (0,0)*{}="A"; (30,0)*{}="B"; (10,25)*{}="C";
    "A"; "B"; **\dir{-}; ?(.33)*{\bullet}="ABB"; ?(.67)*{\bullet}="AAB"; 
    "B"; "C"; **\dir{-};
    "C"; "A"; **\dir{-};
    "C"; "AAB"; **\dir{--}; 
    "C"; "ABB"; **\dir{--}; ?(.4)*{\bullet}="X1";
    "X1"; "AAB"; **\dir{.};
    \endxy};
  (70,0)*{\xy
    (0,0)*{}="A"; (30,0)*{}="B"; (10,25)*{}="C";
    "A"; "B"; **\dir{-}; ?(.33)*{\bullet}="ABB"; ?(.67)*{\bullet}="AAB"; 
    "B"; "C"; **\dir{-};
    "C"; "A"; **\dir{-};
    "C"; "AAB"; **\dir{--}; ?(.5)*{\bullet}="X2"; 
    "C"; "ABB"; **\dir{--}; ?(.4)*{\bullet}="X1";
    "X1"; "AAB"; **\dir{--};
    "X1"; "X2"; **\dir{.};
    \endxy};
\endxy\]
  Note that in the last two cases, cutting along the dotted line reduces the number of almost-free triangles by one, but it increases the total number of triangles. This might make a small difference in performance.
\end{proof}

\begin{lemma}\label{lemma: almost free -> 4 free triangles}
  Every almost-free triangle can be cut into four free triangles.
\end{lemma}

\begin{proof}
  If no vertex is in the conic, the triangle is already free (first case of Remark~\ref{remark: types of almost-free triangles}). If the vertex opposite to the non-free side belongs to the conic (third and fourth cases of Remark~\ref{remark: types of almost-free triangles}) then the segment between them is free, and the triangle is cut into two free pieces.

  There remains only one case (see figure below): the conic intersects the triangle at two points, a vertex $A$ and an interior point $D$ of the side $AB$. The conic cannot intersect $AB$ tangentially (otherwise it would have multiplicity intersection $\geq3$ with that line). Therefore it must enter the triangle through $D$ and it can only exit through $A$.

  When $CD$ is free, this segment cuts the triangle in two free pieces. However this is not true in general.

  Choose any point $P$ in the conic and inside the triangle, with the property that $BP$ and $CP$ are free; then the original triangle is cut in four free pieces.
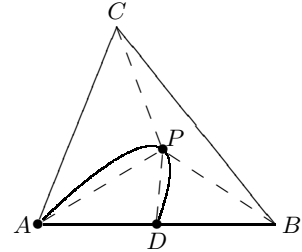
\[\xy
  0; <3pt,0pt>: 
  (0,0)*{\bullet}="A"; (30,0)*{}="B"; (10,25)*{}="C";
  "A"+(-2,0)*{A}; "B"+(2,0)*{B}; "C"+(0,2)*{C};
  "B"; "C"; **\dir{-};
  "C"; "A"; **\dir{-}; ?(.63)*{}="X1";
  "B"; "X1"; **\dir{}; ?(.4)*{\bullet}="P"; "P"+(1.5,1.5)*{P};
  "A"; "B"; **\dir{-}; ?(.5)*{\bullet}="D"; "D"+(0,-2)*{D};
  "A"; "D", {\ellipse:a(45),_:a(70){}};
  "A"; "P"; **\dir{--};
  "B"; "P"; **\dir{--};
  "C"; "P"; **\dir{--};
  "D"; "P"; **\dir{--};
\endxy\]
 It suffices that the tangent to the conic at $P$ leaves $B$ and $C$ on the same half-plane: the branch of the conic must then be contained in the other half-plane, thus $BP$ and $CP$ will be free. We offer three such points which are efficiently computable: the point whose tangent is parallel to $BC$; and the points at which the tangents pass through $B$ or $C$.
\end{proof}

Combining all the previous lemmas and counting the number of pieces at each step, we obtain the following result.

\begin{proposition}
  Every triangle can be cut into eleven free triangles.
\end{proposition}

\begin{remark}
  It is possible to reduce the final number of free pieces to nine but one needs to use more often the recourse of finding tangency points in the conic as in Lemma~\ref{lemma: almost free -> 4 free triangles}; we chose the simpler approach. On the other hand, those points can be computed efficiently, which may make it attractive to minimize the number of triangles in practice. Still, the integration time in each piece depends on the particular intersections.

%
\end{remark}

\subsection{Degenerate conics}\label{section: degenerate conics}

We analyze now how to calculate the integral when $f=0$ is a degenerate conic. If it is empty, one point, or a double line, the integral is zero or the value on the full triangle.

\subsubsection{Two parallel lines}
If $f=0$ is two parallel lines, it can be converted by an affine transformation into $x(x-1)=0$. We can determine in which of the regions $x\leq0$, $0\leq x\leq1$, $1\leq x$ lies the image of each vertex by looking at their $x$-coordinates.
\begin{enumerate}
  \item If all three vertices are in one of the three regions, the integral is either the full triangle integral or zero; we can determine the sign of $f$ in the triangle and use \eqref{eq: full triangle integration formula}.
  \item Otherwise, the triangle is split into two or three pieces (not necessarily triangular). The figure below depicts the possible cases. Once we have determined on which region(s) we must integrate (the middle strip or its complement), this can be done solely by adding and subtracting integrals of triangular pieces, which can be calculated using \eqref{eq: full triangle integration formula}.
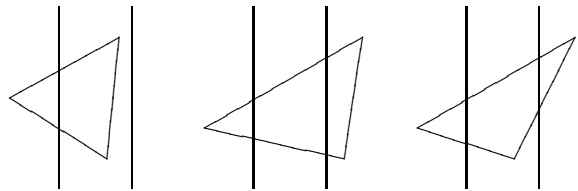
\[\xy
  0; <2.3pt,0pt>: 
  (0,0)*{\xy
    (0,0); (0,30) **\dir{-};
    (12,0); (12,30) **\dir{-};
    (-8,15)*{}="A"; (8,5)*{}="B"; (10,25)*{}="C";
    "A"; "B"; **\dir{-}; 
    "B"; "C"; **\dir{-};
    "C"; "A"; **\dir{-};
  \endxy};
  (35,0)*{\xy
    (0,0); (0,30) **\dir{-};
    (12,0); (12,30) **\dir{-};
    (-8,10)*{}="A"; (15,5)*{}="B"; (18,25)*{}="C";
    "A"; "B"; **\dir{-}; 
    "B"; "C"; **\dir{-};
    "C"; "A"; **\dir{-};
  \endxy};
  (70,0)*{\xy
    (0,0); (0,30) **\dir{-};
    (12,0); (12,30) **\dir{-};
    (-8,10)*{}="A"; (8,5)*{}="B"; (18,25)*{}="C";
    "A"; "B"; **\dir{-}; 
    "B"; "C"; **\dir{-};
    "C"; "A"; **\dir{-};
  \endxy};
\endxy\]
\end{enumerate}

\subsubsection{Two crossing lines}

The conic can be transformed to the pair of lines $xy=0$, what allows us to quickly determine in which quadrants the vertices lie. The region on which to integrate is the intersection of the triangle and two opposing quadrants.
\begin{enumerate}
  \item If all vertices are in the same quadrant, the integral is the full triangle or zero.
  \item If all vertices are in two adjacent quadrants, the triangle is divided in two pieces, one of which is a triangle (or both, if a vertex lies in the limiting line). The integral is that on the triangular piece, or the complementary.
  \item If the triangle is divided in three pieces by the conic, either all vertices are in different regions, or they are in two opposing regions. In any case, we can compute the integral on the relevant region by adding and subtracting integrals on triangles.
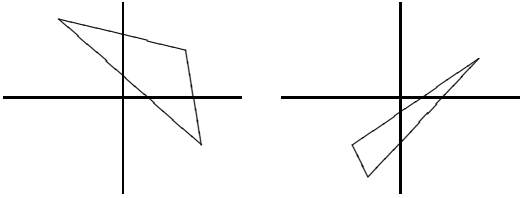
\[\xy
  0; <3pt,0pt>: 
  (0,0)*{\xy
    (-15,0); (15,0) **\dir{-};
    (0,-12); (0,12) **\dir{-};
    (-8,10)*{}="A"; (8,6)*{}="B"; (10,-6)*{}="C";
    "A"; "B"; **\dir{-}; 
    "B"; "C"; **\dir{-};
    "C"; "A"; **\dir{-};
  \endxy};
  (35,0)*{\xy
    (-15,0); (15,0) **\dir{-};
    (0,-12); (0,12) **\dir{-};
    (-4,-10)*{}="A"; (-6,-6)*{}="B"; (10,5)*{}="C";
    "A"; "B"; **\dir{-}; 
    "B"; "C"; **\dir{-};
    "C"; "A"; **\dir{-};
  \endxy};
\endxy\]
  \item Finally, if the triangle is divided in four pieces by the conic, there are two possible arrangements as well. Again, we can compute the integral by adding and subtracting triangles.
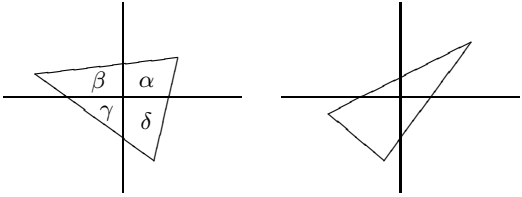
\[\xy
  0; <3pt,0pt>: 
  (0,0)*{\xy
    (-15,0); (15,0) **\dir{-};
    (0,-12); (0,12) **\dir{-};
    (-11,3)*{}="A"; (7,5)*{}="B"; (4,-8)*{}="C";
    "A"; "B"; **\dir{-}; 
    "B"; "C"; **\dir{-};
    "C"; "A"; **\dir{-};
    (3,2)*{\alpha}; (-3,2)*{\beta}; (-2,-2)*{\gamma}; (3,-3)*{\delta};
  \endxy};
  (35,0)*{\xy
    (-15,0); (15,0) **\dir{-};
    (0,-12); (0,12) **\dir{-};
    (-2,-8)*{}="A"; (-9,-2)*{}="B"; (9,7)*{}="C";
    "A"; "B"; **\dir{-}; 
    "B"; "C"; **\dir{-};
    "C"; "A"; **\dir{-};
  \endxy};
\endxy\]
For example, in the left figure, the union of the top-right and bottom-left regions of the triangle is $\alpha+\gamma=(\alpha+\beta+\gamma+\delta)-(\beta+\gamma)-(\delta+\gamma)+2\gamma$, where the terms in parenthesis, as well as $\gamma$, are triangles.
\end{enumerate}
We can decide if we are in situation 1 or 2 by inspecting the signs of the coordinates of the transformed vertices. In order to differentiate situations 3 and 4 we use that the triangle is divided in four pieces if and only if the intersection of the two lines lies in its interior. This can be detected by calculating its barycentric coordinates as in the following algorithm.

\begin{algorithm}\label{algorithm: determine if P is inside ABC}
Determine if a point $P$ is inside a triangle $ABC$.
  \begin{enumerate}
    \item In the expression $\overrightarrow{AP}=\alpha\overrightarrow{AB}+\beta\overrightarrow{AC}$ calculate $\alpha$ and $\beta$:
    \[
      \alpha=\frac{\det(\overrightarrow{AP},\overrightarrow{AC})}{\det(\overrightarrow{AB},\overrightarrow{AC})} , \ \beta=\frac{\det(\overrightarrow{AB},\overrightarrow{AP})}{\det(\overrightarrow{AB},\overrightarrow{AC})}
    \]
    where $\det(u,v)=u_1v_2-u_2v_1$.
    \item If $\alpha,\beta>0$ and $\alpha+\beta<1$ then $P$ is contained in the triangle.
    \item If $\alpha=0$ and $\beta\in[0,1]$; or $\beta=0$ and $\alpha\in[0,1]$; or $\alpha+\beta=1$ and $\alpha\in[0,1]$, then $P$ is in the border of the triangle.
    \item Otherwise $P$ is outside the triangle.
  \end{enumerate}
\end{algorithm}

In any case, for the final integral it is enough to add and substract several instances of \eqref{eq: full triangle integration formula}, with no subdivisions other than the given by the lines of the conic.

\section{Base case integration}\label{section: base case integration}

In this section we describe how to detect the relative position of the nondegenerate conic $f=0$ and a free triangle $ABC$, and compute the integral, in the five possible cases of free triangles.

\subsection{No intersections}

There are three possibilities:
\begin{enumerate}
  \item $T\subset\{f\geq0\}$: the integral was computed in Example~\ref{example: full triangle}.
  \item $T\subset\{f\leq0\}$: the integral is zero.
  \item $f=0$ is an ellipse contained in $T$.
\end{enumerate}
By inspecting $f$ we can decide immediately whether $f=0$ is not an ellipse, from which we would deduce that we are in the first or second case. Then one can discern by evaluating the sign of $f$ at some interior point of the triangle.

On the other hand, if $f=0$ is an ellipse, we have to determine if any of the two shapes is contained in the other. We can do this by mapping the ellipse to the unit circle.

\begin{algorithm}\label{algorithm: determine the relative position of ellipse and triangle and domain of integration}
Determine the relative position of the ellipse $f=0$ and the triangle $ABC$, and the correct domain of integration.
  \begin{enumerate}
    \item Calculate an affine transformation $\phi\colon\mathbb{R}^2\rightarrow\mathbb{R}^2$ that sends $f=0$ into $x^2+y^2=1$. Let $P=(0,0)$.
    \item If $d(\phi(A),P)<1$ then $ABC$ is contained in the ellipse; evaluate the sign of $f$ at some interior point of $ABC$ to decide if the integral is the full triangle or zero.
    \item Otherwise, decide if $P$ is inside the triangle $A'B'C':=\phi(ABC)$ with Algorithm~\ref{algorithm: determine if P is inside ABC}.
    \item If $P$ is in the triangle, then the ellipse is contained in it; evaluate the sign of $f$ at $\phi^{-1}(P)$ to decide on which region to integrate.
    \item Otherwise, none of the shapes contains the other; evaluate the sign of $f$ at $\phi^{-1}(P)$ to decide if the integral is the full triangle or zero.
  \end{enumerate}
\end{algorithm}

The remaining computation is the integral of $g$ when the ellipse $f=0$ is contained in the triangle. We show how to obtain a closed formula when $\{f\geq0\}$ is the bounded region inside the ellipse; in the other case, the required integral is the difference of the full triangle integral and the former.

Let $\varphi=\phi^{-1}\colon\mathbb{R}^2\rightarrow\mathbb{R}^2$ which sends the circle $x^2+y^2=1$ to $f$. Then
\[
  \iint_{\{f\geq0\}} g \, dx\,dy = \iint_D g(\varphi) \, |J(\varphi)| \, dx\,dy
\]
where $D$ is the unit disc. Since $\varphi$ is affine, $|J(\varphi)|\in\mathbb{R}$ and $\overline g:=g(\varphi)$ is again a polynomial. Now, using polar coordinates, this is equal to
\[
  |J(\varphi)| \int_0^{2\pi} \left( \int_0^1 \overline{g}(r\cos\theta,r\sin\theta)\cdot r \ dr \right) d\theta
\]
which is reduced to a linear combination of integrals of type $\int_0^{2\pi} \cos^i\theta \sin^j\theta\,d\theta$.

  Alternatively, by Green's theorem the integral inside the ellipse is
\[
  \iint_E g(x,y) \, dx\,dy = \int_{\partial E} G(x,y)\,dy
\]
where $\frac{\partial G}{\partial x}=g$.

\subsection{One side intersection, no vertex intersections}

This case is entirely similar to the previous one.

\subsection{One vertex intersection}

This case is even simpler: $T$ is contained in $\{f\geq0\}$ or $\{f\leq0\}$, we evaluate the sign of $f$ at some interior point of the triangle in order to decide, and the integral will be that on the full triangle or zero.

\subsection{Two vertex intersections}\label{subsection: two vertex intersections}

This case is more interesting. Either $T$ is contained in one of the regions $\{f\geq0\}$, $\{f\leq0\}$, or it is divided in two regions by the conic. This can be discerned in the following way: determine a segment which cuts the triangle in two (not necessarily triangular) pieces, separating the two relevant vertices, and count the number of intersections of that segment and the conic. Examples: the median of the side determined by the two vertices, or a suitable vertical or horizontal segment. If there are intersections, we are in the latter situation, otherwise evaluate the sign of $f$ inside the triangle to separate the first two possibilities.
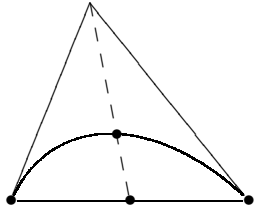
\[\xy
  0; <3pt,0pt>: 
  (0,0)*{\bullet}="A"; (30,0)*{\bullet}="B"; (10,25)*{}="C";
  "B"; "C"; **\dir{-};
  "C"; "A"; **\dir{-};
  "A"; "B"; **\dir{-}; ?(.5)*{\bullet}="D";
  "D"; "C"; **\dir{--}; ?(.665)*{\bullet};
  "A"; "B", {\ellipse:a(-35),_:a(35){}};
\endxy\]
Alternatively, convert the conic to a standard conic and check where the points lie after the transformation (see details in the next subsection).

If the triangle is divided in two regions by the curve, we really have to compute the integral on a region bounded by a conic arc and one or two segments. As usual we can consider only the former (the bottom region in the above picture), without loss of generality. How to determine the actual region of integration? The sign of $f$ in the bottom region is the same as the sign of $f$ in the middle point of the bottom side, for example.

 We can calculate the integral by affinely transforming the conic into a standard conic: the circle $x^2+y^2=1$, the parabola $y=x^2$ or the hyperbola $xy=1$.

\begin{enumerate}
  \item Circle: the integral on the circular segment can be efficiently calculated as the integral on the circular sector minus the integral on the triangle determined by the segment and the center of the circle.

  \item Parabola: the integral after the transform is that on the region $\{y\in[l_{AB}(x),x^2],x\in[a_1,b_1]\}$ where $(a_1,a_2)$ and $(b_1,b_2)$ are the images of the two intersection vertices, with $a_1<b_1$, and $l_{AB}(x)$ is the equation of the line through them.
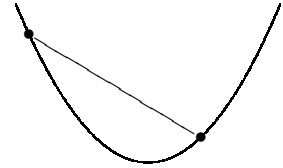
\[\xy
  0; <10pt,0pt>: 
  (5,0)*{}="1"; (-5,0)*{}="2";
  "1"; "2" **\crv{(0,-12)} ?(.97)*\dir{} ?(.95)*{\bullet}="x" ?(.3)*{\bullet}="y";
  "x"; "y"; **\dir{-};
\endxy\]

  \item Hyperbola: similarly to the previous case, the integral can be calculated as that on the region $\{y\in[l_{AB}(x),1/x],x\in[a_1,b_1]\}$ if $a_1<b_1<0$, or $\{y\in[1/x,l_{AB}(x)],x\in[a_1,b_1]\}$ if $0<a_1<b_1$.
\end{enumerate}

\subsection{Three vertex intersections}

As in the previous case, either $T$ is contained in one of the regions $\{f\geq0\}$, $\{f\leq0\}$, or it is divided in two regions by the conic. This time we use a different method to differentiate the three possibilities, since in the third one we also need to know which are the two vertices through which the conic enters the triangle.

Since ellipses and parabolas define a convex region, a triangle with three vertices on such a curve cannot be divided by it. Thus, if the curve is of one of those types, it suffices once more to evaluate the sign of $f$ in the triangle, and calculate the full triangle integral or return zero.

If $f=0$ is a hyperbola, transform $f=0$ into $xy=1$. This curve defines two convex regions, limited by the branches $xy=1,x<0$ and $xy=1,x>0$. By inspecting the signs of the $x$-coordinates of the (transformed) vertices, we can determine in which branch they are.
\begin{enumerate}
  \item If the three vertices are on the same branch of the hyperbola, the triangle is contained in $\{f\geq0\}$ or $\{f\leq0\}$, just determine the sign of $f$ inside.
  \item Otherwise, two vertices lie on one branch and the third vertex lies on the other branch. The integral is calculated as at the end of Section \ref{subsection: two vertex intersections}.
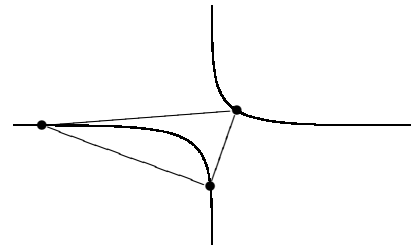
\[\xy
  0; <15pt,0pt>: 
  (-5,0)*{}="-x"; (5,0)*{}="x"; (0,-3)*{}="-y"; (0,3)*{}="y";
  "-y"; "-x" **\crv{(0,0)&} ?(.95)*{\bullet}="A" ?(.2)*{\bullet}="B";
  "x"; "y" **\crv{(0,0)&} ?(.5)*{\bullet}="C";
  "A"; "B" **@{-}; "C" **\dir{-}; "A" **\dir{-};
\endxy\]
\end{enumerate}

Note that the approach used in this case, namely the conversion to a standard conic in order to locate the vertices in relation to the curve, would have worked as well in Section \ref{subsection: two vertex intersections}, when we wanted to decide if the conic separates the triangle in two regions. This would amount to:
\begin{enumerate}
  \item Ellipse: convert to $x^2+y^2=1$ and decide if the third vertex is inside or outside the unit circle.
  \item Parabola: convert to $y=x^2$ and decide if the third vertex is above or below the parabola.
  \item Hyperbola: convert to $xy=1$. If the two intersection vertices have different signs in their $x$-coordinates, the curve cannot separate the triangle. Otherwise, decide if the third vertex is in the convex region limited by the branch where the other two vertices are.
\end{enumerate}

\section{The algorithm}\label{section: the algorithm}

Algorithm \ref{algorithm: main algorithm} (next page) is a compilation of the steps described in the previous sections, so as to present an overview of the complete algorithm. Some case-by-case methods have not been explicitly written for brevity reasons.

\begin{figure*}
\begin{algorithm}\label{algorithm: main algorithm}
Integrate a polynomial $g(x,y)$ of degree 4 on the intersection of a triangle $T$ and the region $\{f\geq0\}$ determined by a quadratic polynomial $f(x,y)$.

\renewcommand{\labelenumii}{\labelenumi\arabic{enumii}.}
\renewcommand{\labelenumiii}{\Alph{enumiii}.}
\renewcommand{\labelenumiv}{\roman{enumiv}.}
  \begin{enumerate}
  \setlength{\parskip}{0pt}
  \item If $C:=\{f=0\}$ is a degenerate conic, go to step 9.
  \item Calculate the intersections of $C$ with each side of $T$.
  \item If all sides of $T$ are not free, let $L:=\{T_1,\ldots,T_n\}$ be a list of free triangular pieces as in Lemma \ref{lemma: no free sides -> 7 free triangles}, and go to step 6.
  \item Otherwise, use Lemma \ref{lemma: 1 free side -> 5 almost-free or free triangles} or Lemma \ref{lemma: 2 free sides -> 4 almost-free or free triangles} to obtain a list $L:=\{T_1,\ldots,T_n\}$ of free or almost-free triangular pieces.
  \item For each triangle in $L$, if it is not free, substitute it in the list by the free pieces provided by Lemma \ref{lemma: almost free -> 4 free triangles}.
  \item Determine the type of $C$.
  \item Initialize $S=0$. For each triangle $T_i$ in $L$:
    \begin{enumerate}
    \item Let $Z_i$ be the intersection of the border of $T_i$ and $C$.
    \item If $Z_i=\emptyset$ or one non-vertex point:
      \begin{enumerate}
      \item If $C$ is an ellipse, use Algorithm \ref{algorithm: determine the relative position of ellipse and triangle and domain of integration} to know the relative position of $C$ and $T_i$.
        \begin{enumerate}
        \item If $C$ is contained in $T_i$, determine the sign of $f$ inside the ellipse. Let $I$ be the integral of $g$ on the bounded region inside $C$, or its complementary with respect to the full triangle, as needed.
        \item In any other case, determine the sign of $f$ inside $T_i$. If it is positive, let $I=\iint_{T_i} g$, otherwise let $I=0$.
        \end{enumerate}
      \item If $C$ is not an ellipse, determine the sign of $f$ inside $T_i$. If it is positive, let $I=\iint_{T_i} g$, otherwise let $I=0$.
      \item Add $I$ to $S$.
      \end{enumerate}
    \item If $Z_i$ is one vertex: determine the sign of $f$ in $T_i$. If it is positive let $I=\iint_{T_i} g$, otherwise let $I=0$. Add $I$ to $S$.
    \item If $Z_i$ is two vertices:
      \begin{enumerate}
      \item Calculate the number of intersections of $C$ with the segment from the middle point of the two vertices to the third vertex.
      \item If there are none, determine the sign of $f$ inside $T_i$. If positive, let $I=\iint_{T_i} g$, otherwise let $I=0$. Add $I$ to $S$.
      \item If there is one, determine which of the two regions is the correct one, by evaluating $f$ in a suitable point.
        \begin{enumerate}
        \item If $C$ is an ellipse, transform it into $x^2+y^2=1$. Calculate the integral on the circular segment. Let $I$ be equal to that value or its complementary with respect to the full triangle.
        \item If $C$ is a parabola, transform it into $y=x^2$. Calculate the integral between the segment and the arc of parabola (the segment is always above). Let $I$ be equal to that value or its complementary with respect to the full triangle.
        \item If $C$ is a hyperbola, transform it into $xy=1$. Calculate the integral between the segment and the arc of hyperbola (which one is above depends on which branch the vertices are in). Let $I$ be equal to that value or its complementary with respect to the full triangle.
        \item Add $I$ to $S$.
        \end{enumerate}
      \end{enumerate}
    \item If $Z_i$ is three vertices:
      \begin{enumerate}
      \item If $C$ is an ellipse or a parabola, determine the sign of $f$ inside $T_i$. If it is positive, let $I=\iint_{T_i} g$, otherwise let $I=0$.
      \item If $C$ is a hyperbola, transform it into $xy=1$ and determine in which branch does each vertex lie.
        \begin{enumerate}
        \item All in one branch: determine the sign of $f$ inside $T_i$. If it is positive, let $I=\iint_{T_i} g$, otherwise let $I=0$.
        \item Two vertices $A,B$ in one branch and the third vertex in the other branch: calculate the integral between the segment $AB$ and the arc of hyperbola (which one is above depends on which branch the vertices are in). Determine the sign of $f$ in the middle point of $AB$. If positive, let $I$ be equal to the calculated integral; if negative, to its complementary with respect to the full triangle.
        \end{enumerate}
      \item Add $I$ to $S$.
      \end{enumerate}
    \end{enumerate}
  \item Output $S$ and stop.
  \item Determine the type of degenerate conic.
    \begin{enumerate}
    \item If $C$ is empty, one point, or a double line, determine the general sign of $f$. If it is positive, let $S=\iint_T g$, otherwise let $S=0$. Output $S$ and stop.
    \item Otherwise, if $C$ is two parallel lines, convert it to $x^2-x=0$; if $C$ is two crossing lines, convert it to $xy=0$.
    \item Determine the position of the vertices with respect to the lines by examining the coordinates of their images by the transformation.
      \begin{enumerate}
      \item If all three vertices are in one of the regions, determine the sign of $f$ inside $T$. If it is positive, let $S=\iint_T g$, otherwise let $S=0$. Output $S$ and stop.
      \item Otherwise, determine the region(s) of integration by evaluating the sign of $f$ at some vertex not on the conic. Write the region of integration as a sum of triangles with $\pm1$ coefficients. Calculate the integral according to this. Output the result and stop. (A case by case method can be easily written.)
      \end{enumerate}
    \end{enumerate}
  \end{enumerate}
\end{algorithm}
\end{figure*}

\subsection{Practical considerations}

In relation to our implementation of this algorithm in MATLAB (almost complete as of May 2010) we would like to comment on numerical aspects that are not considered in our discussion above. First, several transformations suggested (Example~\ref{example: full triangle} and the various transformations into standard conics from Section \ref{section: base case integration}) are a source of rounding errors because for small regions the scaling needed is very large. This problem can be solved by avoiding all scalings, i.e. restricting the transformations to rotations and translations, not to a particular standard conic but to a member of some family of them. The result is a slight complication in the integration formulas, but nothing of concern in terms of efficiency.

An additional problem is that in some cases (the calculation suggested in Section \ref{subsection: two vertex intersections} for the ellipse; Section \ref{section: degenerate conics}) the sought integral is calculated as the difference of two easy integrals which may be orders of magnitude larger than the target, requiring much more precision in order not to lose significant digits.

\section{Application: an optimal control problem}\label{section: application}

Many technical processes are described by partial differential equations. Here, it is important
to optimize these processes. This leads to optimization problems in an infinite-dimensional setting.
As an prototype, we consider the minimization of a convex and quadratic functional subject to a linear elliptic
partial differential equation and inequality constraints on the control.
Let us briefly introduce the optimal control problem we have in mind.

Let $\Omega\subset\mathbb{R}^2$ be a bounded domain with $C^3$-boundary $\Gamma$.
For brevity, we will use $\xi=(x,y)$ to denote points in $\mathbb R^2$.
Let us introduce the following elliptic equation
\begin{equation}\label{eq2p1}
 \begin{aligned}
  -\nabla \cdot (D(\xi) \nabla u(\xi))+c(\xi) u(\xi)&= \chi_{\Omega'}f(\xi) &&\textrm{in }\Omega, \\
                             u(\xi)&= 0 &&\textrm{on }\Gamma. \\
 \end{aligned}
\end{equation}
Here, the control is denoted by $f$, while the solution $u$ of this system is the corresponding state.
Thanks to the assumptions below, for each control $f\in L^2(\Omega)$ there exists a unique response $u\in H^1_0(\Omega)$, which
is a weak solution of equation \eqref{eq2p1}, see e.g.\@ \cite[Sect. 5.8]{gilbargtrudinger}.
The control acts on a compact polygonal subset $\Omega'\subset\Omega$.
Now, we consider the control problem
of minimizing
\begin{equation}\label{eq2p2}
J(f,u)=\frac{1}{2} \iint_{\Omega} (u(\xi)-u_d(\xi))^2 d\xi
+\frac{\alpha}{2} \iint_{\Omega'} f^2(\xi)d\xi
\end{equation}
over all $f\in L^2(\Omega)$
subject to the elliptic equation \eqref{eq2p1} and the control constraints
\begin{equation}\label{eq2p3}
f_a \leq f(\xi)\leq f_b \quad\textrm{a.e.\@ on } \Omega.
\end{equation}
That means, we want find a control $f$ whose response $u$ minimizes the distance to some desired state $u_d$.
Let us denote this optimal control problem \eqref{eq2p1}--\eqref{eq2p3} by (P).
The set of admissible controls for (P) is given by
\[
 F_{ad}=\{f\in L^2(\Omega):\ f_a \leq f\leq f_b \quad \textrm{a.e.\@ on } \Omega\}.
\]

\subsection{Existence and regularity of solutions}

Concerning the data of the state equation \eqref{eq2p1}, we make the following smoothness assumption on the data.

\begin{assumption}\label{ass51}
The coefficients in the differential operator satisfy $D\in C^{1,1}(\bar\Omega)$ and $c\in C^{0,1}(\bar\Omega)$.
Moreover, we assume that $D(x)\geq D_0>0$ and $c(x)\geq 0$ for all $x\in\bar\Omega$.
\end{assumption}

In order to obtain existence of solutions to (P) as well as a-priori discretization error estimates,
we take the following assumptions on the data of the optimization problem.

\begin{assumption}\label{ass52}
We have $\alpha>0$, $u_d\in H^1(\Omega)$, and $f_a$, $f_b \in \mathbb R$
with $f_a\le f_b$ a.e.\@ on $\Omega$.
\end{assumption}

Due to convexity, the problem under consideration is uniquely solvable, with solution denoted by $(u^*,f^*)$.
Moreover, the solution can be characterized by the following necessary optimality conditions.
These conditions are also sufficient since the optimal control problem is convex, see e.g.\@ \cite[Ch. 2]{tro05}.
\begin{theorem}
 Let $f^*$ be the solution of (P) with associated state $u^*$. Then there exists an adjoint state $p^* \in H^1(\Omega)$
 such that the adjoint equation
\begin{equation}\label{eq2p4}
 \begin{aligned}
  -\nabla \cdot (D(\xi) \nabla p^*(\xi))+c(\xi) p^*(\xi)&=(u^*{-}u_d)(\xi) &&\mbox{in }\Omega, \\
                             p^*(\xi)&= 0 &&\mbox{on }\Gamma
 \end{aligned}
\end{equation}
and the variational inequality
\begin{equation}
  \label{eq2p6}
  \iint_{\Omega'}
   (\alpha f^*(\xi)+p^*(\xi))(f(\xi)-f^*(\xi))d\xi \geq  0 \quad \forall f \in F_{ad}
\end{equation}
are satisfied.
Moreover, the following pointwise representation of the optimal control holds
\begin{equation}
  \label{eq2p5}
  f^*(\xi)=\mathcal{P}_{[f_a,f_b]}\left(- \frac{1}{\alpha} {p^*}(\xi)\right) \quad\mbox{a.e.\@ on } \Omega'.
\end{equation}
Here, $\mathcal{P}_{[f_a,f_b]}(f)$ denotes the projection of $f\in \mathbb R$ on the interval $[f_a,f_b]$.
\end{theorem}
Using the projection representation of the optimal control, we can conclude higher regularity of the solution:
\begin{theorem}
 Under the smoothness assumptions \ref{ass51} and \ref{ass52}, it holds $u^*,p^*\in H^3(\Omega)$, $f^*\in H^1(\Omega)$.
\end{theorem}
\begin{proof}
Since we have $p^*\in H^1(\Omega)$ by the previous theorem, the projection representation \eqref{eq2p5}
implies that the optimal control has the same regularity $f^*\in H^1(\Omega)$. Then the right-hand sides of \eqref{eq2p1} and
\eqref{eq2p4} are functions in $H^1(\Omega)$. Standard regularity results for elliptic partial differential equations, e.g.\@ \cite[Thm. 8.13]{gilbargtrudinger}, yield
$u^*,p^*\in H^3(\Omega)$.
\end{proof}

\subsection{Discretization and error estimate}

Now, we turn to the discretization of (P). To that end, let us introduce a family of quasi-uniform
triangulations of $\Omega$, denoted by $\{\mathcal T_h \}_{h>0}$. Each triangulation is assumed to exactly fit
the boundary of $\Omega$, such that $\bar\Omega = \cup_{T\in \mathcal T_h} T$. This implies
that  elements of $\mathcal T_h$ lying on the boundary are curved. We further assume that for each $T\in \mathcal T_h$ there
is a mapping $\Phi_T$ mapping the standard simplex $\hat T$ to $T$. Moreover, we require that the intersection of every triangle $T\in \mathcal T_h$ with the
boundary of the control domain $\Omega'$ is empty. That is, the boundary of $\Omega'$ in $\Omega$ is completely resolved by edges of triangles.

With a triangulation we associate the following space of functions
\[
 V_h= \{ v\in C(\bar\Omega):\ \Phi_T(v|_T)\in P_2(\hat T) \ \forall T\in \mathcal T_h\},
\]
which implies that functions $v_h\in V_h$ are polynomials of degree $2$ on each triangular element.
Since $\Omega'$ is a compact subset of $\Omega$, there is a mesh size $h_0>0$ such that all elements $T\in \mathcal T_h$ with
$\Omega'\cap T\ne\emptyset$ are triangular. Hence, the above developed integration procedure can be applied for functions $v_h\in V_h$ with support in $\Omega'$.

Then the discrete optimal control problem can be written as: minimize $J(u_h,f_h)$ subject to $u_h\in V_h$, $f_h\in F_{ad}$
\begin{equation}\label{eq2p10}
\iint_\Omega\left(D \nabla u_h\nabla v_h + cu_h v_h \right)d\xi = \iint_{\Omega'} f_hv_h d\xi \quad \forall v_h\in V_h.
\end{equation}
Note that we did not explicitly require $f_h$ to be in a finite-dimensional subspace. Nevertheless, if $(u_h^*,f_h^*)$ is a solution of the
discrete problem, there exists a discrete adjoint state $p_h^*\in V_h$ satisfying
\begin{equation}\label{eq2p11}
\iint_\Omega\left( D \nabla p_h^*\nabla v_h + cp_h^* v_h\right) d\xi = \iint_{\Omega} (u_h^*-u_d)v_h d\xi \quad \forall v_h\in V_h
\end{equation}
and
\begin{equation}\label{eq2p12}
 f_h^*=\mathcal P_{[f_a,f_b]}\left(-\frac1\alpha p_h^*\right).
\end{equation}
Due to this projection representation, the control is implicitly discretized as the truncation of a function from the
finite-dimensional space $V_h$.
\begin{theorem}\label{theo55}
Let $(u_h^*,f_h^*,p_h^*)$ be the solution of the discretized optimality system \eqref{eq2p10}--\eqref{eq2p12}.
Then there is a constant $c>0$ independent of the mesh size $h$ such that
\[
\|f_h^*-f^*\|_{L^2(\Omega)} + \|u_h^*-u^*\|_{H^1(\Omega)} + \|p_h^*-p^*\|_{H^1(\Omega)} \le c\ h^3.
\]
\end{theorem}
\begin{proof}
Due to the approximation results of \cite[Ch. 5.4]{brennerscott}, we have that the Assumption 2.4 in \cite{hin05} is satisfied
with $Z=H^3(\Omega)\cap H^1_0(\Omega)$ and convergence order $h^3$.
Then the claim follows by a direct application of \cite[Thm. 2.4]{hin05}.
\end{proof}
Known estimates for piecewise linear elements yield a convergence order of $h^2$ only, compare  \cite{hin05}.
In the two-dimensional case, i.e.\@ $\Omega\subset\mathbb R^2$, the number of unknowns $N=2\,\mbox{dim}V_h$ in the discretized problem
is proportional to $h^{-2}$.
Hence, our result implies that the approximation error is proportional to $N^{-3/2}$,
whereas the use of linear polynomials only reduces the error like $N^{-1}$. This clearly shows  that for optimal control problems
as considered here, the use of piecewise quadratic approximations is preferable.

\subsection{Solution method}

In order to substitute $f_h$ in \eqref{eq2p11} by the projection \eqref{eq2p12}, integrals
\[
\iint_{\{-\alpha^{-1}p_h<f_a\}} f_av_h d\xi,\
\iint_{\{f_a\le-\alpha^{-1}p_h\le f_b\}} p_hv_h d\xi
\]
have to be evaluated for piecewise quadratic polynomials $v_h\in V_h$. This means, any solution method for the discretized problem encounters the difficulties
of integrating over regions bounded by triangles and conics.

The system consisting of the equation \eqref{eq2p10}--\eqref{eq2p12} can be
solved by means of a semi-smooth Newton method, see e.g. \cite{hin05}. Within each step of the method, the
non-smooth equation \eqref{eq2p12} is replaced by a linearized version
\[
 \chi_{\{-\alpha^{-1} p_h^k \in[f_a,f_b]\}} \left( f_h -\frac1\alpha p_h\right)=0 \qquad \mbox{ on }\Omega',
\]
where $p_h^k$ is the adjoint state given by the previous step, and $\chi_A$ denotes the characteristic function of a set $A$.
Multiplying this equation by a test function $v_h\in V_h$ and integrating on $\Omega'$, we obtain
\[\begin{split}
 0&= \iint_{\Omega' \cap \{-\alpha^{-1} p_h^k \in[f_a,f_b]\}}  \left(f_h -\frac1\alpha p_h\right)v_hd\xi  \\
 &= \sum_{T\in \mathcal T_h,\,T\cap\Omega'\ne\emptyset} \iint_{T \cap \{-\alpha^{-1} p_h^k \in[f_a,f_b]\}}  \left(f_h -\frac1\alpha p_h\right)v_hd\xi
\end{split}\]
for all $v_h\in V_h$.
Here, it is important to be able to evaluate the integrals
\[
\iint_{T \cap \{-\alpha^{-1} p_h\in[f_a,f_b]\}} \mathcal P_{[f_a,f_b]}\left(-\frac1\alpha p_h\right) v_h d\xi
\]
and
\[
\iint_{T \cap \{-\alpha^{-1} p_h\in[f_a,f_b]\}} f_h v_h d\xi ,
\]
which can be transformed to the type in the previous sections.

\section{Acknowledgments}

The authors would like to thank J. Schicho for his suggestions and the participants of the Rastenfeld workshop for their feedback.

%
%

\bibliographystyle{abbrv}
\bibliography{\jobname}  

\begin{thebibliography}{1}

\bibitem{brennerscott}
S.~C. Brenner and L.~R. Scott.
\newblock {\em The mathematical theory of finite element methods}, volume~15 of
  {\em Texts in Applied Mathematics}.
\newblock Springer, New York, third edition, 2008.

\bibitem{gilbargtrudinger}
D.~Gilbarg and N.~S. Trudinger.
\newblock {\em Elliptic partial differential equations of second order}.
\newblock Springer-Verlag, Berlin, 1983.

\bibitem{hin05}
M.~Hinze.
\newblock A variational discretization concept in control constrained
  optimization: the linear-quadratic case.
\newblock {\em J. Computational Optimization and Applications}, 30:45--63,
  2005.

\bibitem{tro05}
F.~Tr{\"o}ltzsch.
\newblock {\em {O}ptimale {S}teuerung partieller {D}ifferentialgleichungen}.
\newblock Vieweg, Wiesbaden, 2005.

\bibitem{walker1978}
R.~J. Walker.
\newblock {\em Algebraic curves}.
\newblock Springer-Verlag, New York, 1978.
\newblock Reprint of the 1950 edition.

\end{thebibliography}
%

\end{document}